\documentclass[12pt]{amsart}

\usepackage{amsthm}
\usepackage{amsmath} \usepackage{amsfonts}
\usepackage{amssymb} \usepackage{latexsym} \usepackage{enumerate}
\usepackage{mathtools}
\mathtoolsset{showonlyrefs}
\usepackage{todonotes}

\usepackage{mathrsfs}
\usepackage[numbers]{natbib}

\usepackage{enumerate}
\usepackage{mathtools}

\usepackage[numbers]{natbib}

\renewcommand{\hat}{\widehat}

\usepackage[a4paper, total={6in, 8in}]{geometry}

\numberwithin{equation}{section}

\theoremstyle{plain}%
\newtheorem{Theorem}{Theorem}[section] %
\newtheorem{Lemma}[Theorem]{Lemma} %

\theoremstyle{definition}%

\theoremstyle{remark}%
\newtheorem{Remark}[Theorem]{Remark} %

\renewcommand{\mathcal}{\mathscr}
\newcommand{\cA}{\mathcal{A}}

\newcommand{\cC}{\mathcal{C}}
\newcommand{\cD}{\mathcal{D}}
\newcommand{\cE}{\mathcal{E}}
\newcommand{\cF}{\mathcal{F}}
\newcommand{\cG}{\mathcal{G}}

\newcommand{\cN}{\mathcal{N}}

\newcommand{\cQ}{\mathcal{Q}}

\newcommand{\cT}{\mathcal{T}}

\newcommand{\EE}{\mathbb{E}}

\newcommand{\PP}{\mathbb{P}}
\newcommand{\QQ}{\mathbb{Q}}
\newcommand{\RR}{\mathbb{R}}

\renewcommand{\epsilon}{\varepsilon}

\newcommand{\esssup}{\operatorname*{\mathrm{ess\,sup}}}
\newcommand{\essinf}{\operatorname*{\mathrm{ess\,inf}}}
\newcommand{\essinfext}[1][\P]{#1\text{\rm--}\essinf}
\newcommand{\esssupext}[1][\P]{#1\text{\rm--}\esssup}

\renewcommand{\P}{\PP}

\newcommand{\E}{\EE}

\begin{document}

\title{Optimal investment with a noisy signal of future stock prices}

\author{Peter Bank}
\address{Technische Universit{\"a}t Berlin,
    Institut f{\"u}r Mathematik, Stra{\ss}e des 17. Juni 136, 10623
    Berlin, Germany} 
\email{bank@math.tu-berlin.de}

\author{Yan Dolinsky}
\address{Department of Statistics, Hebrew University of Jerusalem, Mount Scopus, Israel} 
\email{yan.dolinsky@mail.huji.ac.il}

\thanks{P. Bank is supported in part by the
GIF Grant 1489-304.6/2019.\\
Y. Dolinsky is supported in part by the GIF Grant 1489-304.6/2019 and the ISF grant
230/21.}
\date{\today}
\maketitle
\begin{abstract}
We consider an investor who is dynamically informed about the future evolution of one of the independent Brownian motions driving a stock's price fluctuations. With linear temporary price impact the resulting optimal investment problem with exponential utility turns out to be not only well posed, but it even allows for a closed-form solution. We describe this solution and the resulting problem value for this stochastic control problem with partial observation by solving its convex-analytic dual problem.
\end{abstract}
\begin{description}
\item[Mathematical Subject Classification (2010)] 91G10, 91B16
\item[Keywords] optimal control with partial observation, exponential utility maximization, noisy price signals, duality, temporary price impact
\end{description}

\renewcommand{\theequation}{\arabic{section}.\arabic{equation}}
\pagenumbering{arabic}

\section{Introduction}\label{sec:1}

Information is arguably a key driver in  stochastic optimal control problems as policy choice strongly depends on assessments of the future evolution of a controlled system and on how these can be expected to change over time. This is particularly obvious in Finance where information drives price fluctuations and where the question how it affects investment decisions is still a topic of great interest for financial economics and mathematics alike. 

In part this may be due to the many ways information can be modeled and due to the various difficulties that come with each such model choice. If, as we are setting out to do in this paper, one seeks to understand how imperfect signals on future stock price evolution can be factored into present investment decisions, one needs to specify how the signal and the stock price are connected and work out the trading opportunities afforded by the signal. For a dynamic specification of the signal and its noise such as ours, this often leads to challenging combinations of filtering and optimal control and we refer to \cite{B:82, FP:82, B:92, T:98, Betal:19} for classical and more recent work in this direction.

The present paper contributes to this literature with a simple Bachelier stock price model driven by two independent Brownian motions. The investor dynamically gets advanced knowledge about the evolution of one of these over some time period, but the noise arising from the other Brownian motion makes perfect prediction of stock price fluctuations impossible. Accounting in addition for temporary price impact from finite market depth (cf.~\cite{Kyle:85}) as suggested in the Almgren-Chriss model \cite{AlmgrenChriss:01}, we  arrive at a viable expected utility maximization problem. In fact, for exponential utility we follow an approach similar to the one of~\cite{BDR:22}, where the case of perfect stock price predictions is considered, and we work out explicitly the optimal investment strategy and maximum expected utility in this case. We can even allow for dynamically varying horizons over which peeks ahead are possible, thus extending beyond the fixed horizon setting of \cite{BDR:22} (which is still included as a special case). 

Our model also complements the literature on optimal order execution where
different forms of market signals have been studied. For instance, \cite{CarteaJaimungal:16} add a Markovian drift to the stock price evolution to model market trends as perceived by the agent; latent factors are added in \cite{CasgrainJaimungal:19}. With transient price impact rather than temporary one as in our setting, \cite{LehalleNeumann:22} and \cite{NeumannVoss:22} study such signals and \cite{Belaketal:18} adds a version with non-Markovian finite variation signals. Yet another version of market signals is considered in \cite{BankCarteaKoerber:23} who use Meyer $\sigma$-fields to model ultra short-term signals on the jumps in the order flow.

Describing how signals are specifically processed into strategies is a challenge for each of these models. In our setting, we find that the investor should optimally weigh her noisy signals of future stock price evolution and aggregate these weights into a projection of its fluctuations. The difference of this projection with the present stock price determines the signal-based part of her investment strategy. The relevant projection weights are computed explicitly and we analyze how they depend on the signal to noise ratio. On top of that, the investor also is interested in chasing the risk premium afforded by the stock and thus aims to take her portfolio to the Merton ratio, well-known from standard expected utility maximization.  In addition, we investigate the value of reductions in noise and show that it is convex in signal quality. In other words, incentives for signal improvement turn out to be skewed in favor of investors who are already good at sniffing out the market's evolution. Finally, we also use the flexibility in the signal horizon of our model to investigate when our investor would most like to be peek ahead if she can do so only once and how much of an advantage a continually probing investor has over her more relaxed counterpart who does so only periodically.

Mathematically, our method is based on the general duality result given in \cite{BDR:22}. Here, however, we need to understand how changes of measures can be optimized when two Brownian motions are affected. Fortunately, we again find the problem to reduce to deterministic variational problems, which, albeit more involved, remain quadratic and thus can still be solved explicitly. The time-dependent knowledge horizon adds further challenges, but also affords us the opportunity to show that the problem is continuous with respect to perturbations in signal reach. 

The present paper generalizes~\cite{BDR:22} and, while similar in its line of attack, the method of solution here differs in essential ways. Indeed, both papers apply the dual approach and do not deal with the primal infinite dimensional
stochastic control problem directly. But the current paper decomposes the information flow into independent parts and passes to a conditional model driven by independent Brownian motions in the `usual' way. This approach does not require applying the results from the theory of Gaussian Volterra integral equations \cite{H:68, HH:93} and readily accommodates time-dependent signal specifications and noise.

Section~\ref{sec:2} of the paper formalizes the optimal investment problem with dynamic noisy signals mathematically. Section~\ref{sec:3} states the main results on optimal investment strategy and utility and comments on some financial-economic implications. Section~\ref{sec:4} is devoted to the proof of the main results.

\section{Problem formulation}\label{sec:2}

Consider an investor who can invest over some time horizon $[0,T]$
both in the money market at zero interest (for simplicity) and in
stock. The price of the stock follows Bachelier dynamics
\begin{align}
  \label{eq:1}
  S_t = S_0 + \mu t + \sigma\left(\gamma
  W'_t+\bar{\gamma}W_t\right), \quad t \in [0,T],
\end{align}
where $S_0 \in \RR$ denotes the initial stock price, $\mu \in \RR$
describes the stock's risk premium and $\sigma \in (0,\infty)$ its
volatility; $\gamma \in [-1,1]$ and $\bar{\gamma}:=\sqrt{1-\gamma^2}$
parameterize the correlation between the stock price process $S$ and
its drivers  $W$, $W'$, two independent Brownian motions
specified on a complete probability space $(\Omega,\cF,\P)$. On top of
the information flow from present and past stock prices as captured by
the augmented filtration $(\cF^S_t)_{t \in [0,T]}$ generated by
$S$, the investor is assumed to have access to a signal which at any
time $t \in [0,T]$ allows her to deduce the future evolution of $W'$
(but \emph{not} of $W$) over a time window $[t,\tau(t)]$ where
$\tau:[0,T] \to [0,T]$ is a right-continuous, 
nondecreasing time
shift satisfying $\tau(t) \geq t$ throughout. In other words, the
investor is able to partially predict the evolution of future 
stock prices, albeit with some uncertainty. The uncertainty is described by a noise  whose variance accrues over $[t,\tau(t)]$ solely from
$W$ at the rate $\sigma^2(1-\gamma^2)$ and accrues
from both $W$ and $W'$ at the joint rate $\sigma^2$ afterwards. As a
result, the investor can draw on the information flow given by the
filtration 
\begin{align}
  \label{eq:2}
  \cG_t := \cF^S_t \vee \sigma\left(W'_u, u \in [0,\tau(t)]\right),
  \quad t \in [0,T],
\end{align}
when making her investment decisions.

In case $\gamma=\pm 1$ all noise is wiped out from the stock price signal
and so the investor gets perfect knowledge of some future prices, affording
her obvious arbitrage opportunities. But even for the complementary
case $\gamma \in (-1,1)$ it is easy to check that $S$ is not a semimartingale and, by the Fundamental Theorem of Asset Pricing of \cite{DelbaenSchachermayer94}, there is a free lunch with vanishing risk (and in fact even a strong arbitrage as can be shown by a Borel-Cantelli argument similar to the one in~\cite{LeventhalSkorokhod95}).
As a consequence, we need to curb
the investor's trading capabilities in order to maintain a viable financial model. In line with the economic view of the role and effect of arbitrageurs, we
choose to accomplish this by taking into account the market impact
from the investor's trades. These cause execution prices for
absolutely continuous changes $d\Phi_t = \phi_t dt$ in the investor's position to be given by
\begin{align}
  \label{eq:3}
  S^\phi_t := S_t+\frac{\Lambda}{2}\phi_t, \quad t \in [0,T].
\end{align}
So, when marking to market her position $\Phi_t = \Phi_0 + \int_0^t
\phi_s ds$ in the stock accrued by time $t \in [0,T]$, the
investor will consider her net profit  to be
\begin{align}
  \label{eq:4}
  V^{\phi,\Phi_0}_t& := -\int_0^t S^\phi_s d\Phi_s + \Phi_tS_t-\Phi_0S_0 \\&=
  \Phi_0(S_t-S_0)+\int_0^t
  \left(\phi_s(S_t-S_s)-\frac{\Lambda}{2}\phi^2_s\right)ds, \quad t
  \in [0,T].
\end{align}
As a consequence, the investor will have to choose her turnover rates
from the class of admissible strategies 
\begin{align}
  \label{eq:5}
  \cA:=\left\{\phi=(\phi_t)_{t \in [0,T]} \;:\;\phi \text{ is }
  \cG\text{-optional with } \int_0^T \phi^2_t dt<\infty \text{ a.s.}\right\}.
\end{align}
Assuming for convenience constant absolute risk aversion $\alpha \in
(0,\infty)$, the investor would then seek to
\begin{align}
  \label{eq:6}
  \text{Maximize } \E\left[-\exp\left(-\alpha V^{\phi,\Phi_0}_T\right)\right]
  \text{ over } \phi \in \cA.
\end{align}

\begin{Remark}
  The special case of no signal noise ($\gamma=1$) and constant peek
  ahead period (i.e., $\tau(t) = (t+\Delta) \wedge T$, $t \in [0,T]$, for some
  $\Delta>0$) was solved in \cite{BDR:22}.
\end{Remark}

\section{Main results and their financial-economic discussion}\label{sec:3}

The paper's main results are collected in the following theorem.

\begin{Theorem}\label{thm:1}
  The investor's unique optimal strategy is to average out the risk-premium adjusted stock price estimates
  \begin{align}
    \label{eq:7}
    \hat{S}_{t,h}:= \E[S_{t+h}\;|\;\cG_t]-\mu \bar{\gamma}^2h = S_t+\mu \gamma^2 h +\sigma \gamma(W'_{(t+h)
    \wedge \tau(t)}-W'_t), \quad t,h \in [0,\infty), 
  \end{align}
  to obtain the risk- and liquidity-weighted projection of prices
   \begin{align}
     \label{eq:9}
     \bar{S}_t := \frac{1}{\Upsilon_t(\tau(t)-t)}\left(\int_{0}^{\tau(t)-t} \hat{S}_{t,h}
     \Upsilon'_t(h)dh+\Upsilon_t(0)\hat{S}_{t,\tau(t)-t}\right), \quad t \in [0,T], 
   \end{align}
  where
  \begin{align}
    \label{eq:11}
    \Upsilon_t(h)&:=\bar{\gamma}
                   \cosh\left(\bar{\gamma}\sqrt{\rho}h\right)
    +\tanh\left(\sqrt{\rho}(T-\tau(t))\right)\sinh\left(\bar{\gamma}\sqrt{\rho}h\right),
    \quad h \in [0,\infty),
  \end{align}
  with $\rho:=\alpha\sigma^2/\Lambda$.
  With the projection $\bar{S}_t$ at hand, the investor should then 
  take at time $t \in [0,T]$ into view her present
  stock holdings $\hat{\Phi}_t=\Phi_0+\int_0^t \hat{\phi}_sds$  and
  choose to turn her position over at the rate
  \begin{align}
    \label{eq:12}
    \hat{\phi}_t =
    \frac{1}{\Lambda}(\bar{S}_t-S_t)+\frac{\Upsilon'_t(\tau(t)-t)}{\Upsilon_t(\tau(t)-t)}\left(\frac{\mu}{\alpha \sigma^2}-\hat{\Phi}_t\right).
  \end{align}

  Finally, the maximum utility
  the investor can expect is
  \begin{align}
    \label{eq:13}
    &\sup_{\phi \in \cA}
                     \E\left[-\exp\left(-\alpha V^{\phi,\Phi_0}_T\right)\right]\\
                   &= -\exp\left(\frac{\alpha\Lambda\sqrt{\rho}}{2\coth(\sqrt{\rho}T)}\left(\Phi_0-\frac{\mu}{\alpha \sigma^2}\right)^2-\frac{1}{2}\frac{\mu^2}{\sigma^2}T\right)\\&\quad \cdot \exp\left(-\frac{1}{2}\int_{0}^T\!
\frac{\gamma^2\sqrt{\rho}}
{\overline{\gamma}\coth\left(\overline{\gamma}\sqrt{\rho}\left(t-\tau^{-1}(t)\right)\right)+
                     \tanh \left(\sqrt{\rho}\left(T-t\right)\right)}d t \right),
  \end{align}
  where $\tau^{-1}(t):=\inf\{s \in [0,T]:\tau(s)> t\}$, $t\in [0,T]$, denotes the right-continuous inverse of $\tau$.
\end{Theorem}

Let us briefly collect some financial-economic observations from this result. First, similar to the case without signal noise discussed in \cite{BDR:22}, an average of future price proxies $\hat{S}_{t,h}$ is formed, but here weights are time-inhomogenous. In fact, the weights used in this averaging are determined in terms of the function $\Upsilon$ and interestingly depend on the peek-ahead horizon $\tau$ only through its present value $\tau(t)$. As a consequence, when deciding about the time $t$ turnover rate $\hat{\phi}_t$, the investor does not care if her information horizon $\tau$ is going to stall or to jump all the way to $T$ soon. The reason for this is the investor's exponential utility which makes investment decisions insensitive to present wealth and thus does not require planning ahead. It would be interesting to see how this changes with different utility functions. Unfortunately, our method  to obtain an explicit solution to the optimal investment problem  strongly depends on our choice of exponential utility, leaving this a challenge for future research.  

Second, it is interesting to assess the impact of signal noise on investment decisions. Here, we observe that the emphasis $\Upsilon_t(0)/\Upsilon_t(\tau(t)-t)$ that the projection $\bar{S}_t$ puts on the last learned signal $\hat{S}_{t,\tau(t)-t}$ decreases when noise gets stronger due to an increase in $\overline{\gamma}$. This reflects the investor's reduced trust in this most noisy of her signals when noise becomes more prominent. 

Third, we can assess how the presence of noise affects the value of the signal $W'$. To this end, let us consider the certainty equivalent
\begin{align}\label{certain}
    c(\gamma)&:=\frac{1}{\alpha} \log \frac{\sup_{\phi \in \cA }
                     \E\left[-\exp\left(-\alpha V^{\phi,\Phi_0}_T\right)\right]}{\sup_{\phi \in \cA}
                     \E\left[-\exp\left(-\alpha \tilde{V}^{\phi,\Phi_0}_T\right)\right]}\\
                     &= \frac{1}{2\alpha}\int_{0}^T\!
\frac{\gamma^2\sqrt{\rho}}
{\overline{\gamma}\coth\left(\overline{\gamma}\sqrt{\rho}\left(t-\tau^{-1}(t)\right)\right)+
                     \tanh \left(\sqrt{\rho}\left(T-t\right)\right)}d t 
\end{align}
that an agent, who can eliminate a fraction $\gamma^2$ from the variance in stock price noise $\gamma W'+\sqrt{1-\gamma^2}W$ by observing part of $W'$, gains compared to her peer without that privilege (whence her terminal wealth $\tilde{V}_T$ is determined by $\tilde{S}_t=S_0+\mu t + \sigma W_t$, $t\in [0,T]$, instead of $S$). Contributions to this quantity from times long before the investment horizon ($T-t$ large) when the signal $W'_t$ has been received for a long time ($t-\tau^{-1}(t))$ large), depend on $\gamma$ by a factor of about $\gamma^2/(\sqrt{1-\gamma^2}+1)$. This factor increases from 0 to 1 as $\gamma$ increases from 0 (no signal) to 1 (noiseless signal), with the steepest increase at $\gamma=1$, indicating limited, but ever higher returns from noise reductions. Conversely, an increase in noise is making itself felt the most when the signal is already quite reliable. 

Finally, the flexibility concerning the form of the peek-ahead length given by $\tau:[0,T] \to [0,T]$ affords one to account for periods when predictions are harder to make over periods where they are easier. Moreover, this flexibility also allows us to shed light on some natural questions such as the following ones:
\begin{enumerate}
    \item[(i)] \emph{When best to peek ahead?} Suppose that the investor has an opportunity to choose a (deterministic) moment of time
    $\mathbb T\in [0,T]$ when, for one time only, she can peek ahead over some period $\Delta>0$ 
    into the future. What time $\mathbb T$ should she choose? We can formalize this as an optimization problem over the family of time changes 
     \[
 \tau_{\mathbb T}(t) = \left.
  \begin{cases}
    t, & \text{for } t \leq \mathbb T \\
    (t\vee ( \mathbb T+ \Delta)) \wedge T, & \text{otherwise}\\
    \end{cases}
  \right\},\ \ \  \mathbb T\in [0,T].
 \]
    Clearly, the corresponding inverse function is given by 
    \[
 \tau^{-1}_{\mathbb T}(t) = \left.
  \begin{cases}
   \mathbb T, & \text{for } \mathbb T\leq t\leq\mathbb T+\Delta\\ 
     t, & \text{otherwise}\\
    \end{cases}
  \right\}.
 \]
 Hence, 
 in view of 
    \eqref{certain}  we need to maximize over $\mathbb T\in [0,T]$ the value of  
    $$\int_{\mathbb T}^{(\mathbb T+\Delta)\wedge T}
\frac{dt}
{\overline{\gamma}\coth\left(\overline{\gamma}\sqrt{\rho}\left(t-\mathbb T\right)\right)+
                     \tanh \left(\sqrt{\rho}\left(T-t\right)\right)}\rightarrow\max.$$
Observe that 
\begin{align*}
&\int_{\mathbb T}^{(\mathbb T+\Delta)\wedge T}
\frac{dt}
{\overline{\gamma}\coth\left(\overline{\gamma}\sqrt{\rho}\left(t-\mathbb T\right)\right)+
                     \tanh \left(\sqrt{\rho}\left(T-t\right)\right)}\\
                     &=\int_{0}^{(T-\mathbb T)\wedge \Delta}
\frac{dt}
{\overline{\gamma}\coth\left(\overline{\gamma}\sqrt{\rho}t\right)+
                     \tanh \left(\sqrt{\rho}\left(T-\mathbb T-t\right)\right)}\\
                     &\leq
   \int_{0}^{T\wedge \Delta}
\frac{dt}
{\overline{\gamma}\coth\left(\overline{\gamma}\sqrt{\rho}t\right)+
                     \tanh \left(\sqrt{\rho}\left(T-t\right)\right)}\\                  
\end{align*}
We conclude that the optimal time is $\mathbb T=0$ and so one should peek ahead immediately.

\item[(ii)] \emph{How much of an advantage does continual probing have over periodic probing?} For this we can compare the continual probing where we keep predicting what is going to happen over the next $\Delta$ time units with its discrete counter part where we only peek ahead every $\Delta$ units:
\begin{align*}
    \tau^c(t)=(t+\Delta) \wedge T \quad \text{vs.} \quad \tau^p(t)=\left((\lfloor \frac{t}{\Delta} \rfloor+1)\Delta \right) \wedge T,
\end{align*}
where $\lfloor 2.8 \rfloor=2$ denotes the floor function. Clearly, the continual probing is superior to periodic probing, but its advantage depends on the probing period length $\Delta$. Indeed, Figure~\ref{fig:enter-label} below shows the difference in certainty equivalents for the choices $\Delta=T/n$ corresponding to $n=1,2,\dots$ updates in the periodic case. We see that, for the considered parameters, the periodically updating investor will be at the greatest disadvantage compared to her more diligent, continually probing counterpart when updating about 80 times for periods of length $\Delta=T/80$. 
\begin{figure}[h]
    \centering
    \includegraphics[width=.6\linewidth]{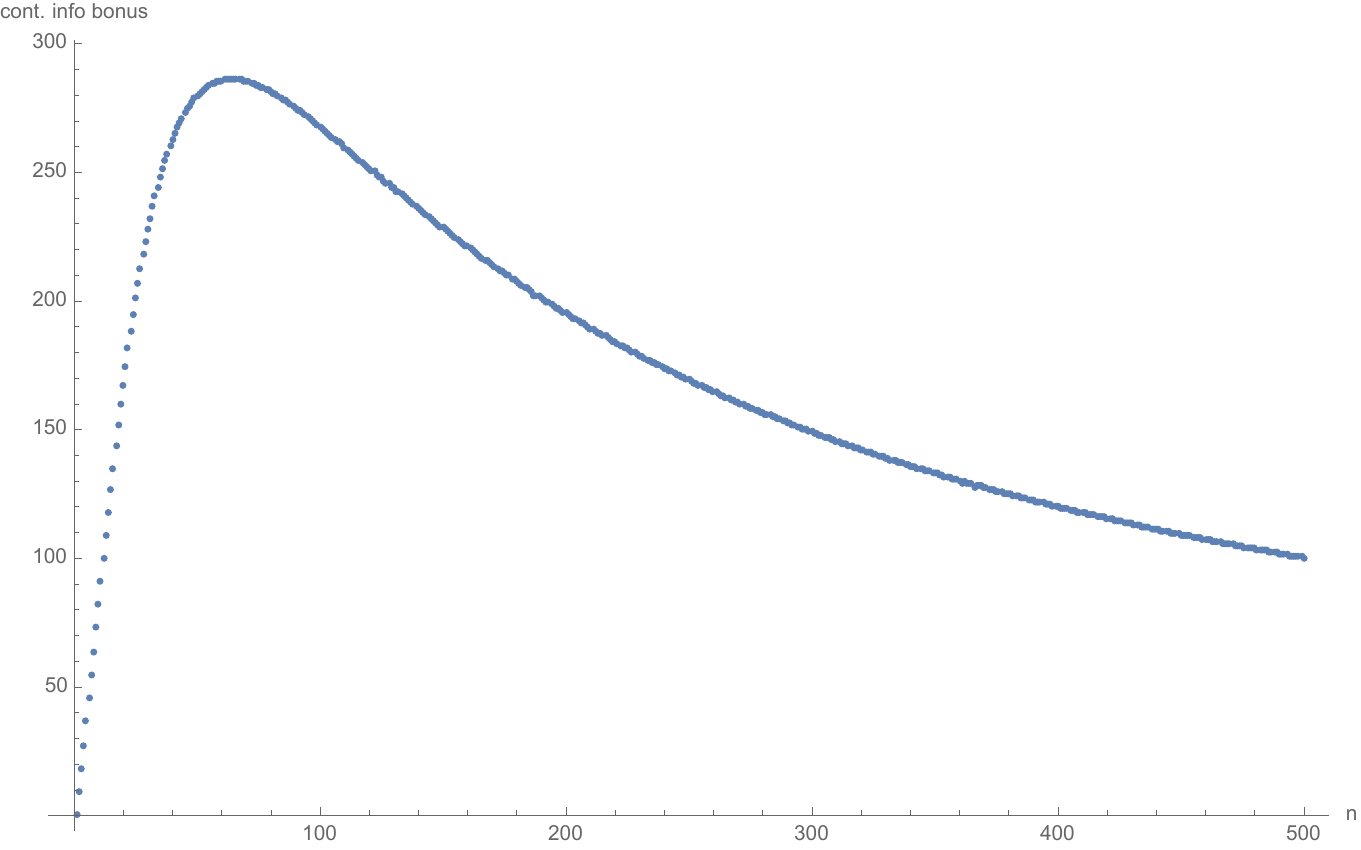}
    \caption{Continual vs.\ periodic probing: Difference of certainty equivalents for $\Delta=T/n$, $n=1,2,\dots,500$, when $\gamma=0.7$, $\rho=0.49$, $\alpha=0.01$, $T=100$.}
    \label{fig:enter-label}
\end{figure}
\end{enumerate}

\section{Proof}\label{sec:4}
We start this section by outlining the main steps of the proof. 
The first step in Section 4.1 deals with simplifying the general setup to obtain a notationally and computationally more convenient special case of our problem. The second step (Section 4.2) consists of the main idea which is to apply duality theory.  We show that the dual problem can be reduced to the solution of three separate deterministic variational problems. 
The third step (Section 4.3) is to solve these deterministic problems, which are quite involved and for which we therefore compute only their values and the essential properties of the corresponding optimal control. The fourth step (Section 4.4) combines the second and the third step to solve the dual problem. The last step (Section 4.5) uses duality and the Markovian structure of our problem to finally construct the solution to the primal problem from the dual solution.

\subsection{Simplification of our setting}

\paragraph{\emph{Parameter reduction.}}
Let us first note that it suffices to treat the special case where
$S_0=0$, $\sigma=1$, $\mu=0$, $\alpha=1$:

\begin{Lemma}
  If $\hat{\phi}^0 \in \cA$ is optimal when starting with the initial
  position $\Phi^0_0:=\alpha\sigma \Phi_0-\mu/\sigma$ in the baseline model with
  parameters $S^0_0:=0$, $\sigma^0:=1$, $\mu^0:=0$, $\alpha^0:=1$ and
  $\Lambda^0:=\Lambda/(\alpha \sigma^2)$, then the optimal strategy in
  the model with general parameters $\Phi_0$, $S_0$, $\sigma$, $\mu$, $\alpha$
  and $\Lambda$ is given by
  \begin{align}
    \label{eq:28}
    \hat{\phi}_t
    (\omega)=\frac{1}{\alpha \sigma}\hat{\varphi}_t\left(\left(W_s(\omega)+\frac{\mu}{\sigma}\gamma
    s,W'_s(\omega) +\frac{\mu}{\sigma}\overline{\gamma} s\right)_{s
    \in [0,T]}\right), \quad t \in [0,T],
  \end{align}
  where $\hat{\varphi}$ is the measurable functional on $C[0,T]\times C[0,T]$ for
which $\hat{\phi}^0=\hat{\varphi}(W',W)$.  Morevorer, the maximum
expected utility is
$\exp\left(-\frac{1}{2}\frac{\mu^2}{\sigma^2}T\right)$ times as high
as the maximum utility in the baseline model.
\end{Lemma}
\begin{proof}
 Let $\phi=\varphi(W,W') \in \cA$ be the representation of an arbitrary
 admissible strategy as a functional of the underlying Brownian motions
 $W$ and $W'$. Then $(\tilde{W},\tilde{W}'):=\left(W_s(\omega)+\frac{\mu}{\sigma}\gamma
    s,W'_s(\omega) +\frac{\mu}{\sigma}\overline{\gamma} s\right)_{s
    \in [0,T]}$, 
 is a standard two-dimensional Brownian motion under $\tilde{\P}\approx \P$ with $d\tilde{\P}/d\P=\cE(-\frac{\mu}{\sigma}(\gamma
 W'+\overline{\gamma}W))_T$, $(\tilde{W},\tilde{W}')$ and we can write
 the expected utility from $\phi$ as
 \begin{align}
   \label{eq:155}
   \E[-\exp(-\alpha V_T^{\phi,\Phi_0})]&=\exp\left(-\frac{1}{2}\frac{\mu^2}{\sigma^2}T\right)\tilde{\E}\left[-\exp\left(
                                         -\tilde{V}_T^{\varphi^0_0,\tilde{\phi}}\right)\right]
 \end{align}
 where $\tilde{V}_T^{\Phi^0_0,\tilde{\phi}}$ is the terminal wealth
 generated in the baseline model from the initial position
 $\Phi^0_0$ by the strategy given by
   \begin{align}
    \label{eq:288}
    \tilde{\phi}_t
    (\omega)={\alpha \sigma}{\varphi}_t\left(\left(\tilde{W}(\omega),\tilde{W}'(\omega)\right)_{s \in [0,T]}\right) , \quad t \in [0,T].
  \end{align}
 In the baseline model,
 $\hat{\phi}^0=\hat{\varphi}(\tilde{W},\tilde{W}')$ is optimal by
 assumption and we obtain an upper bound for the maximum utility in
 the general model. This bound is sharp since it is attained by
 $\hat{\phi}$ as in the formulation of the lemma.
\end{proof}

\paragraph{\emph{Reduction to differentiable, increasing time changes.}}
Let us argue next why we can assume without loss of generality that
\begin{align}
  \label{eq:14}
  \tau \text{ is continuously differentiable and strictly increasing
  on } [0,\tau^{-1}(T)],
\end{align}
and that it suffices to do the computation of the value of the
problem~\eqref{eq:13} for $\tau$ satisfying, in addition, $\tau(0)=0$.

In fact, using that the description of optimal policies from
Theorem~\ref{thm:1} holds for $\tau$ satisfying~\eqref{eq:14} and
that~\eqref{eq:13} holds when, in addition, $\tau(0)=0$ (which we will prove independently below), we can even
prove that the optimization problem depends continuously on $\tau$ and
thus it suffices to consider smooth $\tau$ with~\eqref{eq:14}:

\begin{Lemma}
  If $\tau_n(t)\to \tau_\infty(t)$ for $t=T$ and for every continuity
  point $t$ of $\tau_\infty$ (i.e., if $d\tau_n \to d\tau_\infty$ weakly as Borel
  measures on $[0,T]$), then the optimal terminal wealth associated
  with $\tau_n$ converges almost surely to the one associated with
  $\tau_\infty$ and the problem value for $\tau_n$ converges to the one of
  $\tau_\infty$.
\end{Lemma}
\begin{proof}
   Let us denote by $v(\tau)$ the value of our problem in dependence
   on $\tau$ and by $\tilde{v}(\tau)$ its claimed value from the
   right-hand side of~\eqref{eq:13}. Similarly, denote by $V(\tau)$
   the optimal terminal wealth (if it exists) for $\tau$ and by
   $\tilde{V}(\tau)$ our candidate described in
   Theorem~\ref{thm:1}. Let us furthermore denote by $\cT$ the class
   of all $\tau$ satisfying~\eqref{eq:14} and by $\cT_0 \subset \cT$
   those $\tau$ which in addition satisfy $\tau(0)=0$. We will assume
   that $V(\tau)=\tilde{V}(\tau)$ for $\tau \in \cT$ and
   $v(\tau)=\tilde{v}(\tau)$ for $\tau \in \cT_0$ as these statements will be
   derived below independently from this lemma.

   Let us observe first that both $\tilde{V}(\tau)$ and
   $\tilde{v}(\tau)$ continuously depend on $\tau$. For $\tilde{v}$
   this readily follows from~\eqref{eq:13} by dominated convergence;
   for $\tilde{V}$ this is due to the stability of the linear
   ODE~\eqref{eq:12} and the continuous dependence of its coefficients
   on $\tau$.
     
   Let us next argue why an optimal control for general $\tau$ exists and
   its terminal wealth is $V(\tau)=\tilde{V}(\tau)$. For this take
   $\tau_n \in \cT$ converging to $\tau$ from above and observe
   that $V(\tau_n)=\tilde{V}(\tau_n) \to \tilde{V}(\tau)$.  By Fatou's
   lemma we conclude that
   $\limsup_n \E[-\exp(-\alpha V(\tau_n))] \leq \E[-\exp(-\alpha
   \tilde{V}(\tau))] \leq v(\tau)$. Because $\tau_n \geq \tau$, any
   competitor $\phi \in \cA(\tau)$ is also in $\cA(\tau_n)$ and thus
   $v(\tau_n)=\E[-\exp(-\alpha V(\tau_n))] \geq \E[-\exp(-\alpha
   V_T^{\phi,\Phi_0})] $.  In conjunction with the preceding
   estimate, this shows that the candidate for general $\tau$ is indeed
   optimal with value $v(\tau)=\tilde{v}(\tau)=\lim_n v(\tau_n)$.

  Let us conclude by arguing why it suffices to compute the problem
  value for $\tau \in \cT_0$. For this choose $\tau_n \in \cT_0$ to
  converge to a general $\tau$ from below. From the above, we know
  that $v(\tau) \geq \tilde{v}(\tau)=\lim_n v(\tau_n)$. To see that conversely
  $\tilde{v}(\tau) \geq {v}(\tau)$ (and conclude), take a
  risk-aversion $\alpha'<\alpha$ and observe that the corresponding
  (candidate) problem values $v_{\alpha'}(.)$ and $\tilde{v}_{\alpha'}(.)$ satisfy
  \begin{align}
    \label{eq:29a}
    \tilde{v}_{\alpha'}(\tau) = \lim_n v_{\alpha'}({\tau}_n) 
    \geq \lim_n \E[-\exp(-\alpha' V(\tau_n))] =\E[-\exp(-\alpha' V(\tau))].
  \end{align}
  Indeed, the first identity is due to the stability of the right-hand
  side of~\eqref{eq:13} in $\tau$; the estimate holds because the
  optimal strategy for risk-aversion $\alpha$ is also admissible for
  risk-aversion $\alpha'$; the final identity holds because of uniform
  integrability which in turn follows from boundedness in $L^p(\P)$ with
  $p=\alpha/\alpha'>1$:
  \begin{align}
    \label{eq:32}
    \sup_n \E[\exp(-\alpha' {V}(\tau_n))^p]=\sup_n \E[\exp(-\alpha {V}(\tau_n))] = \sup_n(-\tilde{v}(\tau_n))<\infty.
  \end{align}
  Letting $\alpha' \uparrow \alpha$ in~\eqref{eq:29a}, its left-hand
  side converges to $\tilde{v}(\tau)$ by continuous dependence on
  $\alpha$ of the right-side of~\eqref{eq:13}; the right-hand side
  in~\eqref{eq:29a} converges by monotone convergence to $\E[-\exp(-\alpha
 V(\tau))]=v(\tau)$ and we are done.
\end{proof}

\paragraph{\emph{Decomposing the filtration into independent Brownian parts and passage
  to a conditional model.}}
For time changes $\tau$ satisfying~\eqref{eq:14}, we can introduce 
\begin{align}
  \label{eq:17}
  B_t:=\int_{0}^{t\wedge\tau^{-1}(T)}
  \frac{1}{\sqrt{\dot{\tau}(u)}}dW'_{\tau(u)}, \quad t \in [0,T],
\end{align}
and readily check that it is a Brownian motion stopped at time
$\tau^{-1}(T)\leq T$ which is independent of both
$W$ and $(W'_s)_{s \in [0,\tau(0)]}$. Moreover, using that
$\tau^{-1}(t)=0$ for $t \in [0,\tau(0)]$, we can write the stock
price dynamics as
\begin{align}
  \label{eq:19}
  S_t=\gamma W'_{t\wedge \tau(0)}+\gamma \int_{0}^{\tau^{-1}(t)}
\sqrt{\dot{\tau}(s)}dB_s+\overline{\gamma} W_t, \quad  t\in [0,T],
\end{align}
and view the
insider's filtration as generated by the following independent
components:
\begin{align}
  \label{eq:18}
  \cG_t = \sigma(W'_s, \; s \in [0,\tau(0)]) \vee \sigma(B_s, \; s \in
  [0, t \wedge \tau^{-1}(T)])\vee \sigma(W_s,
  \; s \in [0,t]) \vee \cN,
\end{align}
for $t \in [0,T]$ and with $\cN$ denoting the collection of $\P$-nullsets. So, maximizing
expected utility conditional on $\cG_0=\sigma(W'_s, \; s \in [0,\tau(0)]) \vee \cN$ amounts to maximizing
an \emph{unconditional} expected utility in a model with probability
$\P_0$ where, by a slight abuse of notation, asset prices
evolve according to
\begin{align}
  \label{eq:20}
  S_t:=\gamma w(t\wedge \tau(0))+\gamma \int_{0}^{\tau^{-1}(t)}
\sqrt{\dot{\tau}(s)}dB_s+\overline{\gamma} W_t, \quad  t\in [0,T],
\end{align}
for some deterministic path segment $w \in C([0,\tau(0)])$, a
$\P_0$-Brownian motion $B$ stopped at time $\tau^{-1}(T)$ and a
standard $\P_0$-Brownian motion $W$ which generate the filtration
$$
\cG^\tau_t:=\sigma(B_s, \; s \in
  [0, t \wedge \tau^{-1}(T)])\vee \sigma(W_s,
  \; s \in [0,t]) \vee \cN, \quad t \in [0,T],
$$
specifying the information flow for admissible strategies.

\subsection{Duality}

For the unconditional expected utility maximization under the measure
$\P_0$ identified above, we can apply
Proposition~A.2 in \cite{BDR:22} to deduce that we can proceed by solving the
dual problem with value
\begin{align}\label{eq:dualproblem}
\inf_{\mathbb Q\in\mathcal Q}\mathbb E_{\mathbb Q}\left[\Phi_0(S_T-S_0)+\log \frac{d\mathbb Q}{d\mathbb P_0}+\frac{\rho}{2}\int_{0}^T\left|
\mathbb E_{\mathbb Q}\left[S_T\middle|\mathcal
  G^{\tau}_t\right]-S_t\right|^2 dt\right]
\end{align}
where $\rho:=1/\Lambda$ and $\mathcal Q$ is the set of all probability
measures $\mathbb Q\approx \mathbb P_0$ with finite relative entropy
$\mathbb E_{\mathbb Q}\left[\log\frac{d\mathbb Q}{d\mathbb
    P_0} \right]<\infty$. Said proposition also yields that the unique
solution $\hat{\QQ}$ to the dual problem allows us to construct the
solution $\hat{\phi}$ to the primal problem considered under $\P_0$ as
\begin{align}
  \label{eq:29}
  \hat{\phi}_t = \frac{1}{\Lambda}\E_0[S_T-S_t\;|\;\cG^\tau_t], \quad
  t \in [0,T].
\end{align}

Following the path outlined in the
special case treated in~\cite{BDR:22}, we need to rewrite the dual
target functional~\eqref{eq:dualproblem}. For this it will be convenient to
introduce the functionals 
$\Psi_1:(L^2[0,\tau^{-1}(T)],dt)\times (L^2[0,T],dt)\rightarrow\mathbb R$ and 
$\Psi_2,\Psi_3:L^2\left([0,\tau^{-1}(T)]^2, dt\otimes ds\right)\times L^2\left([0,T]^2, dt\otimes ds\right)\rightarrow\mathbb R$ given by
  \begin{align}
  \Psi_1(a, \tilde a)&:=\frac{1}{2}\int_{0}^{\tau^{-1}(T)}a^2_t dt+\frac{1}{2}\int_{0}^T \tilde a^2_t dt\label{opt1}\\
  &\quad+\Phi_0\left(\gamma
\left(w({\tau(0)})+\int_{0}^{\tau^{-1}(T)}\sqrt{\dot{\tau}(t)}a_t dt\right)
+\overline{\gamma}\int_{0}^T \tilde a_t dt \right)\nonumber\\
&\quad+\frac{\rho}{2}\int_{0}^T\left(\gamma\left( w({\tau(0)})-w({t\wedge\tau(0)})+
\int_{\tau^{-1}(t)}^{\tau^{-1}(T)}\sqrt{\dot{\tau}(s)}a_s ds\right)+\overline{\gamma} \int_{t}^T\tilde a_s ds\right)^2 dt,\\
  \Psi_2(l,\tilde l)&:=\frac{1}{2}\int_{0}^{\tau^{-1}(T)}\int_{s}^{\tau^{-1}(T)}
  l^2_{t,s} dt\, ds+\frac{1}{2}\int_{0}^{\tau^{-1}(T)}\int_{s}^T \tilde
                      l^2_{t,s}dt\, ds\label{opt2}\\
 &\quad +\frac{\rho}{2}\int_{0}^{\tau^{-1}(T)}\int_{s}^T
\left(\gamma \int_{s\vee \tau^{-1}(t)}^{\tau^{-1}(T)}\sqrt{\dot{\tau}(u)} l_{u,s}du+\overline{\gamma}\int_{t}^T \tilde l_{u,s} du+
\gamma\sqrt{\dot{\tau}(s)}\mathbb I_{s\geq\tau^{-1}(t)}\right)^2
dt\, ds\nonumber
\intertext{and}\\
  \Psi_3(m,\tilde m)&:=\frac{1}{2}\int_{0}^{\tau^{-1}(T)}\int_{s}^{\tau^{-1}(T)}
  m^2_{t,s} dt\, ds+\frac{1}{2}\int_{0}^T\int_{s}^T \tilde m^2_{t,s}dt ds\label{opt3}\\
 &+\frac{\rho}{2}\int_{0}^T\int_{s}^T
\left(\gamma\mathbb I_{s\leq\tau^{-1}(T)} \int_{s\vee \tau^{-1}(t)}^{\tau^{-1}(T)}\sqrt{\dot{\tau}(u)} m_{u,s}du+\overline{\gamma}\int_{t}^T \tilde m_{u,s} du\right)^2\nonumber
dt\,ds.
\end{align}
\begin{Lemma}\label{lem:2}
  The dual infimum in~\eqref{eq:dualproblem} coincides with the one taken over
  all $\mathbb Q \in \mathcal Q$ whose densities take the form
  \begin{align}\label{eq:standarddensity}
\left.\frac{d\mathbb Q}{d\mathbb P_0}\right|_{\cG^\tau_t} =
    \cE\left(\int_0^{. \wedge \tau^{-1}(T)} \theta_s
    dB_s+\int_0^. \tilde\theta_s dW_s\right)_t, \quad t \in [0,T],
  \end{align}
  for some bounded, $(\cG^\tau_t)$-adapted $\theta$ and $\tilde\theta$
  changing values only at finitely many deterministic times. For such
  $\mathbb Q$, we have the Girsanov $(\cG^\tau,\QQ)$-Brownian motions
      \begin{align}
      \label{eq:22}
      B^{\mathbb Q}_t &:= B_t - \int_0^t \theta_s ds, \quad t \in [0,\tau^{-1}(T)],\\
      W^{\mathbb Q}_t &:= W_t - \int_0^t \tilde\theta_s ds, \quad t \in [0,T],
    \end{align}
  which allow for Ito representations 
    \begin{align}
      \label{eq:21}
      \theta_t &= a_t + \int_0^{t} l_{t,s}
                 dB^{\mathbb Q}_s+\int_0^t m_{t,s} dW^{\mathbb Q}_s, \quad t \in
                 [0,\tau^{-1}(T)],\\
      \tilde\theta_t &= \tilde{a}_t + \int_0^{t} \tilde l_{t,s}
      dB^{\mathbb Q}_s+\int_0^t \tilde m_{t,s} dW^{\mathbb Q}_s, \quad
                       t \in [0,T],
    \end{align}
    for suitable $a_t, \tilde{a}_t \in \RR$ and suitable
    $\mathcal G^{\tau}$-predictable processes $l_{t,.}$, $m_{t,.}$,
    $\tilde l_{t,.}$, $\tilde m_{t,.}$. In terms of these,
     the dual target value associated with $\QQ$ is
   \begin{align}\label{representation}
   &\mathbb E_{\mathbb Q}\left[\Phi_0(S_T-S_0)+\log \frac{d\mathbb Q}{d\mathbb P_0}+\frac{\rho}{2}\int_{0}^T\left|
\mathbb E_{\mathbb Q}\left[S_T\middle|\mathcal
  G_t\right]-S_t\right|^2 dt\right]\nonumber\\
  &=\mathbb E_{\mathbb Q}\left[\Psi_1(a,\tilde a)+\Psi_2(l, \tilde l)+\Psi_3(m, \tilde m)\right]
  \end{align}
\end{Lemma}
\begin{proof}
For any $\mathbb Q \in \mathcal Q$ the martingale representation
  property of Brownian motion gives us a predictable $\theta,\tilde\theta$ with
  $$\mathbb E_{\mathbb Q}\left[\log\left(\frac{d\mathbb Q}{d\mathbb P}\right) \right]=\frac{1}{2}\mathbb
  E_{\mathbb
    Q}\left[\int_0^{\tau^{-1}(T)}\theta^2_tdt+\int_0^{T}\tilde\theta^2_tdt\right]<\infty$$
  such that~\eqref{eq:standarddensity} holds. Using this density to
  rewrite the dual target functional as an expectation under
  $\mathbb P$, we can use standard density arguments to see that the
  infimum over $\mathbb Q \in \mathcal Q$ can be realized by
  considering the $\mathbb Q$ induced by simple $\theta,\hat\theta$ as
  described in the lemma's formulation
  via~\eqref{eq:standarddensity}. As a consequence, the It\^o
  representations of $\theta_t,\hat\theta_t$ can be chosen in such a
  way that the resulting
  $a_t,\tilde a_t, l_{t,.}\tilde l_{t,.},m_{t,.},\tilde m_{t,.}$ are
  also measurable in $t$; in fact they only change when
  $\theta,\hat\theta$ change their values, i.e., at finitely many
  deterministic times.  This measurability property will be used below
  for applying Fubini’s theorem.

Next, we compute the value of the dual problem in terms of $a,\tilde a,l,\tilde l,m,\tilde m$.
Recalling the dynamics for $S$ in \eqref{eq:19}, we get
\begin{equation}\label{1}
\mathbb E_{\mathbb Q}\left[\Phi_0(S_T-S_0)\right]=\Phi_0\left(\gamma
\left(w(\tau(0))+\int_{0}^{\tau^{-1}(T)}\sqrt{\dot{\tau}(t)}a_t dt\right)
+\overline{\gamma}\int_{0}^T \tilde a_t dt \right).
\end{equation}
From It\^o's isometry and Fubini's theorem we obtain
\begin{align}\label{2}
\mathbb E_{\mathbb Q}\left[\log \frac{d\mathbb Q}{d\mathbb P}\right]&=
\frac{1}{2}\int_{0}^{\tau^{-1}(T)}a^2_t dt+\frac{1}{2}\int_{0}^T \tilde a^2_t dt\\
&\quad+\frac{1}{2}\int_{0}^{\tau^{-1}(T)}\int_{s}^{\tau^{-1}(T)}
  l^2_{t,s} dt ds
  +\frac{1}{2}\int_{0}^{\tau^{-1}(T)}\int_{s}^T \tilde l^2_{t,s}dtds
  \nonumber\\
  &\quad+\frac{1}{2}\int_{0}^{\tau^{-1}(T)}\int_{s}^{\tau^{-1}(T)}
  m^2_{t,s} dt\, ds+\frac{1}{2}\int_{0}^T\int_{s}^T \tilde m^2_{t,s}dt\,ds.\nonumber
\end{align}
Next, from the Fubini Theorem
\begin{align*}
&\mathbb E_{\mathbb Q}\left[S_T-S_t|\mathcal G^\tau_t\right]=\gamma\int_{\tau^{-1}(t)}^{t\wedge \tau^{-1}(T)}\sqrt{\dot{\tau}(s)}dB^{\mathbb Q}_s\\
&+\gamma\left( w({\tau(0)})-w({t\wedge\tau(0)})+
\int_{\tau^{-1}(t)}^{\tau^{-1}(T)}\sqrt{\dot{\tau}(s)}a_s ds\right)+\overline{\gamma} \int_{t}^T\tilde a_s ds\\
&+\gamma\left(\int_{0}^{t\wedge\tau^{-1}(T)}\int_{s\vee \tau^{-1}(t)}^{\tau^{-1}(T)}
\sqrt{\dot{\tau}(u)} l_{u,s} du\, dB^{\mathbb Q}_s+
\int_{0}^{t\wedge\tau^{-1}(T)}\int_{s\vee \tau^{-1}(t)}^{\tau^{-1}(T)}  \sqrt{\dot{\tau}(u)} m_{u,s} du \,d\hat{B}^{\mathbb Q}_s\right)\\
&+\overline{\gamma}\left(\int_{0}^{t\wedge\tau^{-1}(T)}\int_{t}^T \tilde l_{u,s} du\, dB^{\mathbb Q}_s+\int_{0}^t\int_{t}^T \tilde m_{u,s} du\, d\tilde {B}^{\mathbb Q}_s\right).
\end{align*}
This together with the It\^o's isometry and \eqref{1}--\eqref{2} gives
\eqref{representation}.
\end{proof}

\subsection{Solving the deterministic variational problems}

Lemma~\ref{lem:2} above suggests to consider the minimization of the
functionals $\Psi_1$, $\Psi_2$, and $\Psi_3$ specified there. This amounts to solving deterministic variational
problems and the following lemmas sum up the main findings.

    \begin{Lemma}\label{lem:Gamma2}
      The functional $\Psi_2=\Psi_2(l,\tilde{l})$ defined for
      square-integrable $l$, $\tilde{l}$ by \eqref{opt2} has the minimum
      value
   \begin{align}\label{40}
    \Psi_2(\underline{l},\underline{\tilde{l}})&=\frac{1}{2}\int_{\tau(0)}^T\!
\frac{\gamma^2\sqrt{\rho}}
{\overline{\gamma}\coth\left(\overline{\gamma}\sqrt{\rho}\left(t-\tau^{-1}(t)\right)\right)+
                     \tanh \left(\sqrt{\rho}\left(T-t\right)\right)}dt
   \end{align}
 for some $(\underline{l},\underline{\tilde{l}})$.
        \end{Lemma}
\begin{proof}
For a given $(l,\tilde l)\in L^2\left([0,\tau^{-1}(T)]^2, dt\otimes ds\right)\times L^2\left([0,T]^2, dt\otimes ds\right)$
and $s<\tau^{-1}(T)$ introduce the functions
 \begin{align}
&g_s(t):=\int_{\tau^{-1}(t)}^{\tau^{-1}(T)}\sqrt{\dot{\tau}(u)} l_{u,s}du, \ \ t\in [\tau(s),T]\label{func1}\\
&\tilde g_s(t):=\int_{t}^T  \tilde l_{u,s} du, \ \ t\in [s,T].\label{func2}
\end{align}
Clearly, almost surely the functions $g_s,\tilde g_s$ are absolutely
continuous in $t$ for almost every $s \in [0,\tau^{-1}(T)]$ with
boundary values
 \begin{align}\label{boundary}
   g_s(T)=0,
   \quad \tilde g_s(T)=0.
 \end{align}
Changing variables we note
$$
\int_{s}^{\tau^{-1}(T)} l^2_{t,s}dt=\int_{\tau(s)}^T \dot{g}^2_s(t)dt
$$
where $\dot{g}_s$ denotes the derivative with respect to $t$
and so
$$\Psi_2(l,\tilde l)=\int_{0}^{\tau^{-1}(T)} \psi_s(g_s,\tilde{g}_s) ds$$
for
\begin{align}\label{5}
 \psi_s(g,\tilde{g})&:=\frac{1}{2}\left(\int_{\tau(s)}^T \left(\dot{g}^2(t)+\dot{\tilde g}^2(t)\right)dt\right)+
 \frac{1}{2}\int_{s}^{\tau(s)} \dot{\tilde g}^2(t)dt\\
 &\quad +\frac{\rho}{2}\int_{\tau(s)}^T
\left(\gamma g(t)+\overline{\gamma}\tilde g(t)
\right)^2dt\nonumber\\
&\quad+\frac{\rho}{2}\int_{s}^{\tau(s)}\left(\gamma\sqrt{\dot{\tau}(s)}+\gamma g\left(\tau(s)\right)+\overline{\gamma}\tilde g(t)\right)^2 dt.\nonumber
\end{align}

Minimization of $\Psi_2$ can thus be carried out separately for each
$s \in [0,\tau^{-1}(T)]$. So we fix such an $s$ and seek to determine
the absolutely continuous $(\underline{g}_s,\underline{\tilde{g}}_s)$
which minimizes $\psi_s$ under the boundary condition
\eqref{boundary}. 
To this end, observe fist that the functional $\psi_s$
is strictly convex and so existence of a minimizer
follows by a
standard Komlos-argument.

In the computation of the minimum value, let us, for ease of notation, put $g:=g_s$,
$\tilde{g}:=\tilde{g}_s$. We first focus on the contributions to
$\psi_s(g,\tilde{g})$ from its
integrals over the interval $[\tau(s),T]$, assuming $\tau(s)<T$ of
course; for $\tau(s)=T$ we can directly proceed with the minimization
over $[s,\tau(s)]$ as carried out below with $z_s:=0$.
Clearly if $\gamma x+\overline{\gamma} y$ is given then the minimum of $x^2+y^2$ is obtained for
$x,y$ which satisfy
$\frac{y}{x}=\frac{\overline{\gamma}}{\gamma}$.  Hence, the optimal solution will satisfy
$$g(t)=\gamma f_s(t), \ \ \tilde g(t)=\overline{\gamma} f_s(t), \ \ t\in [\tau(s),T]$$
for some function $f_s$ which satisfies $f_s(T)=0$.
Hence the functional to be minimized is really
\begin{align*}
&\frac{1}{2}\int_{\tau(s)}^T \left(\dot{g}_s^2(t)+\dot{\tilde g}_s^2(t)\right)dt+
\frac{\rho}{2}\int_{\tau(s)}^T
\left(\gamma g_s(t)+\overline{\gamma}\tilde g_s(t)
\right)^2
dt\\
&=\frac{1}{2}\int_{\tau(s)}^T \dot{f}^2_s(t)+\frac{\rho}{2}\int_{\tau(s)}^T
f^2_s(t)dt.
\end{align*}
Its Euler–Lagrange reads $\ddot{f_s}=\rho f_s$; c.f., e.g., Section~1 in \cite{GF:63}.
Since $f_s(T)=0$, we conclude that
\begin{align}
\underline{g}_s(t)&=\gamma z_s\frac{\sinh \left(\sqrt{\rho}(T-t)\right)}
{\sinh \left(\sqrt{\rho}\left(T-\tau(s)\right)\right)},
\label{A}\\
\underline{\tilde g}_s(t)&=\overline{\gamma} z_s\frac{\sinh \left(\sqrt{\rho}(T-t)\right)}
{\sinh \left(\sqrt{\rho}\left(T-\tau(s)\right)\right)}, \quad  t\in [\tau(s),T],
\label{B}
\end{align}
for some $z_s \in \RR$ which is yet to be determined optimally.

Next, we treat the interval $[s,\tau(s)]$.
The Euler–Lagrange equation for
minimizing the functional
$$\int_{s}^{\tau(s)}\dot {\tilde g}_s^2(t)dt+\frac{\rho}{2}\int_{s}^{\tau(s)}
\left(\gamma\sqrt{\dot{\tau}(s)}+\gamma^2 z_s+\overline{\gamma}\tilde g_s(t)\right)^2
dt$$
subject to the boundary condition
$\tilde g_s(\tau(s))=\overline{\gamma} z_s$ (which we infer
from~\eqref{B}) is
$$\ddot {\tilde g}_s=\rho\overline{\gamma} \left(\gamma \sqrt{\dot{\tau}(s)}+\gamma  z_s+\overline{\gamma}\tilde g_s\right).$$
Hence,
\begin{align}\label{C}
\underline{\tilde g}_s(t)+&\frac{\gamma}{\overline{\gamma}}\left(\sqrt{\dot{\tau}(s)}+ z_s\right)=
u_s\sinh \left(\overline{\gamma}\sqrt{\rho}\left(\tau(s)-t\right)\right)\\
&+
\frac{1}{\overline{\gamma}}\left(\gamma\sqrt{\dot{\tau}(s)}+ z_s\right)
\cosh \left(\overline{\gamma}\sqrt{\rho}\left(\tau(s)-t\right)\right)
, \quad
t\in [s,\tau(s)],
\end{align}
for some $u_s \in \RR$ which we still need to choose optimally.

We plug \eqref{A}--\eqref{C} into \eqref{5} to obtain 
\begin{align*}
\psi_s(\underline{g}_s,\underline{\tilde{g}}_s)&=\frac{1}{2}\sqrt{\rho}\coth\left(\sqrt{\rho}
\left(T-\tau(s)\right)\right)z_s^2\\
 &\quad +\frac{1}{2\overline{\gamma}}\sqrt{\rho}
 \sinh \left(\overline{\gamma}\sqrt{\rho}\left(\tau(s)-s\right)\right)\cosh \left(\overline{\gamma}\sqrt{\rho}\left(\tau(s)-s\right)\right)\\&\quad \quad \quad\cdot\left(\overline{\gamma}^2u_s^2+\left(\gamma\sqrt{\dot{\tau}(s)}+z_s\right)^2\right)
\\
 &\quad+\rho\sinh^2 \left(\overline{\gamma}\sqrt{\rho}\left(\tau(s)-s\right)\right)
 \left(\gamma\sqrt{\dot{\tau}(s)}+z_s\right)u_s\bigg)
\end{align*}
Simple computations show that, for $\tau(s)<T$, the minimal value of
the above positive definite quadratic pattern in $(u_s,z_s) \in \RR^2$ is equal to
\begin{align*}
{}&\frac{1}{2}\gamma^2\sqrt{\rho}
\frac{
\coth\left(\sqrt{\rho}\left(T-\tau(s)\right)\right)
\tanh\left(\overline{\gamma}\sqrt{\rho}\left(\tau(s)-s\right)\right)}
{\overline{\gamma}\coth\left(\sqrt{\rho}\left(T-\tau(s)\right)\right)+
                 \tanh\left(\overline{\gamma}\sqrt{\rho}\left(\tau(s)-s\right)\right)}\dot{\tau}(s)\\
&=  \frac{1}{2}
\frac{\gamma^2\sqrt{\rho}}
{\overline{\gamma}\coth\left(\overline{\gamma}\sqrt{\rho}\left(\tau(s)-s\right)\right)+
\tanh\left(\sqrt{\rho}\left(T-\tau(s)\right)\right)}\dot{\tau}(s)
\end{align*}
and is attained for some constants $z_s$ and $u_s$ which depend
continuously on~$s$. For $\tau(s)=T$, we have $z_s=0$ and the
minimization over $u_s\in \RR$ leads again to this last expression
which thus holds for any $s$.
Upon integration over $s$, we thus find
\begin{align*}
  \Psi_2(\underline{l},\underline{\tilde{l}})=
&\frac{1}{2}
\int_{0}^T \frac{\gamma^2\sqrt{\rho}}
{\overline{\gamma}\coth\left(\overline{\gamma}\sqrt{\rho}\left(\tau(s)-s\right)\right)+
\tanh\left(\sqrt{\rho}\left(T-\tau(s)\right)\right)}\dot{\tau}(s)
 ds
\end{align*}
and \eqref{40} follows by the substitution $t=\tau(s)$.

\end{proof}

\begin{Lemma}\label{lem:Gamma1}
 The functional $\Psi_1=\Psi_1(a,\tilde{a})$ given for
 square-integrable $(a,\tilde{a})$ by~\eqref{opt1} is minimized by
 $(\underline{a},\tilde{\underline{a}}) :=\Psi\left(\Phi_0,(w(t)-w(0))_{t
   \in [0,\tau(0)]}\right)$ for some
continuous functional $\Psi: \RR\times C([0,\tau(0)])\rightarrow
L^2\left([0,\tau^{-1}(T)], dt\right)\times L^2\left([0,T], dt\right)$.
This minimizer $(\underline{a},\tilde{\underline{a}})$ satisfies
\begin{align}\label{26}
\gamma &w(\tau(0))+\gamma \int_{0}^{\tau^{-1}(T)}\sqrt{\dot{\tau}(t)} \underline{a}_t dt+\overline{\gamma}\int_{0}^T \tilde{\underline{a}}_t dt\\
&=\bar{S}_0-S_0-
\frac{\rho\Upsilon'_0(\tau(0))}{\Upsilon_0(\tau(0))}\Phi_0
\end{align}
where $\bar S$ and $\Upsilon$ are given in Theorem
\ref{thm:1}.

Moreover,
in the special case $\tau(0)=0$ the minimum value is
\begin{align}\label{27}
\Psi_1(\underline{a},\tilde{\underline{a}})=-\frac{1}{2 \sqrt{\rho}}\tanh\left(\sqrt\rho T\right)\Phi^2_0.
\end{align}
\end{Lemma}
\begin{proof}
Similar arguments as in Lemma \ref{lem:Gamma2} give the uniqueness of the minimizer.

For a given $a,\tilde a\in L^2\left([0,\tau^{-1}(T)], dt\right)\times L^2\left([0,T], dt\right)$
introduce the functions
\begin{align}\label{17}
h(t)&:=\int_{\tau^{-1}(t)}^{\tau^{-1}(T)}\sqrt{\dot{\tau}(s)} a_s ds, \quad  t\in [\tau(0),T],\\
 \tilde h(t)&:=\int_{t}^T  \tilde a_s ds, \quad t\in [0,T].\label{18}
\end{align}
Clearly, $h,\tilde h$ are absolutely continuous and satisfy
 \begin{align}\label{boundary1}
 h(T)=\tilde h(T)=0.
 \end{align} 

From a change of variables we obtain
$$
\int_{0}^{\tau^{-1}(T)} a^2_{t}dt=\int_{\tau(0)}^T \dot{h}^2(t)dt
$$
and so $\Psi_1(a,\tilde a)=\psi_1(h,\tilde h)$ for
\begin{align}\label{7}
 \psi_1 (h,\tilde h)&:=\Phi_0\left(\gamma
w({\tau(0)})+\gamma h\left(\tau(0)\right)+\overline{\gamma}\tilde h(0) \right)\\
&\quad+\frac{1}{2}\int_{\tau(0)}^T \left(\dot{h}^2(t)+
\dot{\tilde h}^2(t)\right)dt+\frac{1}{2}\int_{0}^{\tau(0)} \dot {\tilde h}^2(t)dt\\
 &\quad+\frac{\rho}{2}\int_{0}^{\tau(0)}\left(\gamma\left(w({\tau(0)})-w(t)\right)+
\gamma h\left(\tau(0)\right)+\overline{\gamma} \tilde h(t)\right)^2 dt\\
&\quad+\frac{\rho}{2}\int_{\tau(0)}^{T}\left(\gamma h(t)+\overline{\gamma} \tilde h(t)\right)^2 dt.
\end{align}
We aim to find absolutely continuous functions $h,\tilde h$ which
minimize the right hand side of \eqref{7} under the boundary condition
\eqref{boundary1}.  

Using the same arguments as in
Lemma~\ref{lem:Gamma2}, we obtain that a unique minimizer
$(h,\tilde h)$ exists. Focusing first on minimizing the contributions to $\gamma$ collected
over $[\tau(0),T]$ (if $\tau(0)<T$), we moreover find that it satisfies
\begin{align}
h(t)&=\gamma v\frac{\sinh \left(\sqrt{\rho}(T-t)\right)}
{\sinh \left(\sqrt{\rho}\left(T-\tau(0)\right)\right)}, \quad  t\in [\tau(0),T]\label{7+}\\
\tilde h(t)&=\overline{\gamma} v\frac{\sinh \left(\sqrt{\rho}(T-t)\right)}
{\sinh \left(\sqrt{\rho}\left(T-\tau(0)\right)\right)}, \quad t\in [\tau(0),T]\label{7++}
\end{align}
for some $v=h(\tau(0))/\gamma=\tilde{h}(\tau(0))/\overline{\gamma}$ which we still need to
optimize over.

For the contributions to $\psi_1(h,\tilde{h})$ over the interval
$[0,\tau(0)]$, i.e., for the minimization of
\begin{align}
  \label{eq:30}
  \frac{1}{2}\int_{0}^{\tau(0)} \dot {\tilde h}^2(t)dt+\frac{\rho}{2}\int_{0}^{\tau(0)}\left(\gamma\left(w({\tau(0)})-w(t)\right)+
\gamma h\left(\tau(0)\right)+\overline{\gamma} \tilde h(t)\right)^2 dt
\end{align}
over $\tilde{h}$ subject to the yet to be optimally chosen starting value
$\tilde{h}(0)=u$ and the terminal value
$\tilde{h}(\tau(0))=\overline{\gamma}v$ fixed above,
we apply Theorem~3.2
in \cite{BSV:17} for our present deterministic setup. This result
shows that
\begin{align}\label{7+++}
\tilde \xi(t)&:=\frac{1-\gamma^2 \cosh \left(\overline{\gamma}\sqrt{\rho}(\tau(0)-t)\right)}{\overline{\gamma}
\cosh \left(\overline{\gamma}\sqrt{\rho}(\tau(0)-t)\right)}v\\
&\quad+\frac{\gamma\sqrt\rho}{\cosh \left(\overline{\gamma}\sqrt{\rho}\left(\tau(0)-t\right)\right)}\int_{t}^{\tau(0)} \sinh \left(\overline{\gamma}\sqrt{\rho}\left(\tau(0)-s\right)\right)\left(w(s)-w(\tau(0))\right)ds, \quad t\in [0,\tau(0)]
\end{align}
allows us to describe the minimizer $\underline{\tilde{h}}$ as the
solution to the linear ODE
\begin{align}\label{ODE}
\dot{\tilde h}(t)=\overline{\gamma}\sqrt{\rho}\coth
  \left(\overline{\gamma}\sqrt{\rho}(\tau(0)-t)\right)\left(\tilde
  \xi(t)-\tilde{h}(t)\right), \quad \tilde{h}(0)=u,
\end{align}
and the minimum value is
\begin{align}
&\Phi_0\overline{\gamma}\tilde h(0)+\frac{1}{2}\int_{0}^{\tau(0)} \dot {\tilde h}^2(t)dt\nonumber\\
&+
\frac{1}{2\Lambda}\int_{0}^{\tau(0)}\left(\gamma\left(w(\tau(0))-w(t)\right)+
\gamma h\left(\tau(0)\right)+\overline{\gamma} f(t)\right)^2 dt \nonumber\\
=\;&\Phi_0\overline{\gamma}u+\frac{1}{2}\overline{\gamma}\sqrt{\rho}\coth \left(\overline{\gamma}\sqrt{\rho}\tau(0)\right)\left(u-\hat\xi(0)\right)^2\label{20}\\
&+\frac{1}{2\Lambda}\int_{0}^{\tau(0)}\left(\frac{v}{\cosh \left(\overline{\gamma}\sqrt{\rho}\left(\tau(0)-t\right)\right)}+\gamma f(t)\right)^2 dt\nonumber
\end{align}
for
\begin{align}\label{4.21}
&f(t):=w(\tau(0))-w(t)\\
&+\frac{\overline{\gamma}\sqrt{\rho}}{\cosh \left(\overline{\gamma}\sqrt{\rho}\left(\tau(0)-t\right)\right)}
\int_{t}^{\tau(0)} \sinh \left(\overline{\gamma}\sqrt{\rho}\left(\tau(0)-s\right)\right)\left(w(s)-w(\tau(0))\right)ds.\nonumber
\end{align}
Minimizing over $u$ in \eqref{20} we obtain
\begin{equation}\label{8}
\underline{\tilde h}(0)=\hat\xi(0)-\Phi_0\tanh \left(\overline{\gamma}\sqrt{\rho}\tau(0)\right)/\sqrt{\rho}
\end{equation}
and we find the corresponding minimum value to be
\begin{align}\label{9}
&\Phi_0\overline{\gamma} \underline{\tilde h}(0)+\frac{1}{2}\int_{0}^{\tau(0)} \dot {\underline{\tilde h}}^2(t)dt\\
&+
\frac{1}{2\Lambda}\int_{0}^{\tau(0)}\left(\gamma\left(w(\tau(0))-w(t)\right)+
\gamma h\left(\tau(0)\right)+\overline{\gamma} \underline{\tilde h}(t)\right)^2 dt\\
=\;&\Phi_0\overline{\gamma}\hat\xi(0)-
\frac{1}{2 \sqrt{\rho}}\Phi^2_0 \overline{\gamma}\tanh \left(\overline{\gamma}\sqrt{\rho}\tau(0)\right)\\
&+\frac{1}{2\Lambda}\int_{0}^{\tau(0)}\left(\frac{v}{\cosh \left(\overline{\gamma}\sqrt{\rho}\left(\tau(0)-t\right)\right)}+\gamma f(t)\right)^2 dt.\nonumber
\end{align}
Next,  we plugin \eqref{7+}--\eqref{7++} and \eqref{9} into \eqref{7}. The result is
\begin{align}
  \psi_1(\underline{h},\underline{\tilde{h}})&=\Phi_0\gamma w(\tau(0))+\Phi_0\gamma^2 v+\frac{1}{2}\sqrt{\rho}\coth\left(\sqrt{\rho}(T-\tau(0)\right)v^2\label{21}\\
 &\quad+\Phi_0\overline{\gamma}\hat\xi(0)-
   \frac{1}{2 \sqrt{\rho}}\Phi^2_0\overline{\gamma}\tanh
   \left(\overline{\gamma}\sqrt{\rho}\tau(0)\right) \\
 &\quad+\frac{\rho}{2}\int_{0}^{\tau(0)}\left(\frac{v}{\cosh \left(\overline{\gamma}\sqrt{\rho}\left(\tau(0)-t\right)\right)}+\gamma f(t)\right)^2 dt
\end{align}
and from (\ref{7+++}) (for $t=0$) it follows that the minimum over $v$  is achieved
for
\begin{align}\label{29}
v:=-\sqrt{\rho}\frac{\frac{\Phi_0}{\rho\cosh\left(\overline{\gamma}\sqrt{\rho}
\tau(0)\right)}+\gamma\int_{0}^{\tau(0)}\frac{f(t)}{\cosh \left(\overline{\gamma}\sqrt{\rho}\left(\tau(0)-t\right)\right)}dt}
{\coth\left(\sqrt{\rho}
\left(T-\tau(0)\right)\right)+\frac{1}{\overline{\gamma}}\tanh\left(\overline{\gamma}\sqrt{\rho}
\tau(0)\right)}.
\end{align}
The minimizer $(\underline{a},\underline{\tilde{a}})$ can now be recovered from
\eqref{17}--\eqref{18},
\eqref{7+}--\eqref{ODE}, \eqref{4.21} and \eqref{29} and turns out to
be a continuous functional of $(\Phi_0, (w(t)-w(0))_{t \in
  [0,\tau(0)]})$. From
\eqref{7+}, \eqref{7+++}, \eqref{8}
\begin{align*}
&\gamma w(\tau(0))+\gamma h(\tau(0))+\overline{\gamma} \tilde h(0)\\
&=\gamma w(\tau(0))+
\frac{\nu}{
\cosh \left(\overline{\gamma}\sqrt{\rho}\tau(0)\right)}-
\overline{\gamma}\Phi_0\tanh \left(\overline{\gamma}\sqrt{\rho}\tau(0)\right)/\sqrt{\rho}\\
&\quad+\frac{\gamma\overline{\gamma}\sqrt\rho}{\cosh \left(\overline{\gamma}\sqrt{\rho}\tau(0)\right)}\int_{0}^{\tau(0)} \sinh \left(\overline{\gamma}\sqrt{\rho}\left(\tau(0)-s\right)\right)\left(w(s)-w(\tau(0))\right)ds\\
&=\bar{S}_0-S_0-\Lambda
\frac{\Upsilon'_0(\tau(0))}{\Upsilon_0(\tau(0))}\Phi_0
\end{align*}
where the last equality follows from \eqref{4.21}, \eqref{29}
and the Fubini theorem.
This together with \eqref{17}--\eqref{18} gives \eqref{26}.
Finally, \eqref{27} follows from
\eqref{21}--\eqref{29}.
\end{proof}

\begin{Lemma}\label{lem:Gamma3}
The functional $\Psi_3=\Psi_s(m,\tilde{m})$ given by~\eqref{opt3} is
minimized by $(\underline{m},\underline{\tilde{m}})=(0,0)$ which
  yields the minimum value $\Psi_3(\underline{m},\underline{\tilde{m}})= 0$. 
    \end{Lemma}
    \begin{proof}
    The statement is obvious since $\Psi_3\geq 0$ and
    $\Psi_3(m,\tilde m)=0$ if and only if $(m,\tilde m)=(0,0)$.
    \end{proof}

\subsection{Construction of the solution to the dual problem}

Having found the minimum values for the functionals $\Psi_1$,
$\Psi_2$, and $\Psi_3$ from Lemma~\ref{lem:2}, we conclude that their
sum yields a lower bound for the value of the dual
problem~\eqref{eq:dualproblem}. In fact, this lower bound is sharp since
we can construct a measure $\underline{\QQ} \in \cQ$ whose dual target value
coincides with it. For this consider the resolvent
kernel $\underline{k}$ associated with $\underline{l}$ via
\begin{align}
  \label{eq:24}
  \underline{l}_{t,s}+\underline{k}_{t,s} = \int_s^t\underline{l}_{t,r}\underline{k}_{r,s}dr,
  \quad 0 \leq s \leq t \leq T.
\end{align}
This kernel allows us to construct the solution
\begin{align}
  \label{eq:8}
  \underline{\theta}_t := \underline{a}_t-\int_0^t
  \underline{\kappa}_{t,r}\left(\underline{a}_r+\int_0^r\underline{l}_{r,s}dB_s\right)dr,
  \quad t \in [0,T],
\end{align}
to  the Volterra equation
\begin{align}
  \label{eq:23}
  \underline{\theta}_t = \underline{a}_t+\int_0^t
  \underline{l}_{t,s}dB_s + \int_0^t \underline{l}_{t,s} \underline{\theta}_sds.
\end{align}
Now, the probability $\QQ=\underline{\QQ}$ given
by~\eqref{eq:standarddensity} with $\theta=\underline{\theta}$ of
\eqref{eq:8} and
\begin{align}
  \label{eq:26}
  \tilde{\theta}=\underline{\tilde{\theta}}:=\underline{\tilde{a}}+\int_0^. \underline{\tilde{l}}_{.,s}(dB_s+\underline{\theta}_sds).
\end{align}
induces the Girsanov $\underline{\QQ}$-Brownian motion
$B^{\underline{\QQ}}$ of \eqref{eq:22} and, using that
$\underline{\theta}$ solves the Volterra
equation~\eqref{eq:23}, we find the Ito representations
\begin{align}
  \label{eq:25}
  \underline{\theta}_t &= \underline{a}_t + \int_0^t
  \underline{l}_{t,s}dB^{\underline{\QQ}}_s \text{ for  } t \in
                         [0,\tau^{-1}(T)], \\
  \underline{\tilde{\theta}}_t&=\underline{\tilde{a}}_t+\int_0^t \underline{\tilde{l}}_{t,s}dB^{\underline{\QQ}}_s
                              \text{ for } t \in [0,T].
\end{align}
As a consequence, $\underline{\theta}$, $\underline{\tilde{\theta}}$,
  $\underline{a}$, $\underline{\tilde{a}}$, $\underline{l}$,
  $\underline{\tilde{l}}$ and $\underline{m}=0$,
  $\underline{\tilde{m}}=0$ are related exactly in the same way as the
  corresponding quantities in Lemma~\ref{lem:2}. By the same calculation as
  in the proof of this lemma, the dual target value thus turns out to
  be
     \begin{align}
   &\mathbb E_{\underline{\mathbb Q}}\left[\Phi_0(S_T-S_0)+\log \frac{d\underline{\mathbb Q}}{d\mathbb P_0}+\frac{\rho}{2}\int_{0}^T\left|
\mathbb E_{\underline{\mathbb Q}}\left[S_T\middle|\mathcal
  G_t\right]-S_t\right|^2 dt\right]\nonumber\\
       &=\Psi_1(\underline{a},\tilde{\underline{a}})+\Psi_2(\underline{l},\underline{\tilde{l}})+\Psi_3(\underline{m},\underline{\tilde{m}}),
     \end{align}
where in the last line we are allowed to drop the expectation
$\E_{\underline{\QQ}}$ because the quantity there is deterministic anyway.
It follows that $\underline{\QQ}$ indeed attains the lower bound for the dual
problem.

Moreover, we learn from Lemmas~\ref{lem:Gamma1}, \ref{lem:Gamma2},
and~\ref{lem:Gamma3}, that the dual problem's value is given
by
\begin{align}
\inf_{\mathbb Q\in\mathcal Q}&\mathbb E_{\mathbb Q}\left[\Phi_0(S_T-S_0)+\log \frac{d\mathbb Q}{d\mathbb P_0}+\frac{\rho}{2}\int_{0}^T\left|
\mathbb E_{\mathbb Q}\left[S_T\middle|\mathcal
                               G^{\tau}_t\right]-S_t\right|^2 dt\right]\\
                             &=\Psi_1(\underline{a},\tilde{\underline{a}})+\Psi_2(\underline{l},\underline{\tilde{l}})+\Psi_3(\underline{m},\underline{\tilde{m}})\\
  &= -\frac{1}{2 \sqrt{\rho}}\tanh\left(\sqrt{\rho}T\right)\Phi^2_0\\
  &\qquad +\frac{1}{2}\int_{0}^T\!
\frac{\gamma^2\sqrt{\rho}}
{\overline{\gamma}\coth\left(\overline{\gamma}\sqrt{\rho}\left(t-\tau^{-1}(t)\right)\right)+
                     \tanh \left(\sqrt{\rho}\left(T-t\right)\right)}dt,
\end{align}
where the last identity holds in case
$\tau(0)=0$.

\subsection{Construction of the solution to the primal problem}

Having constructed the dual optimizer $\hat{\QQ}:=\underline{\QQ}$, we now can use~\eqref{eq:29} to
compute the optimal primal solution $\hat{\phi}$. For $t=0$, we use the representation of $S_T$ from~\eqref{eq:19} and the fact that $B$ and $W$ have, respectively, drift $\underline{\theta}$ and $\underline{\tilde{\theta}}$ under $\underline{\QQ}$. Hence, upon taking expectation with respect to $\underline{\QQ}$, we find
\begin{align}
    \E_{\underline{\QQ}}\left[S_T\right]&=\E_{\underline{\QQ}}\left[\gamma w(\tau(0))+\gamma\int_0^{\tau^{-1}(T)} \sqrt{\dot{\tau}(t)}\theta_t dt+\overline{\gamma}\int_0^T\underline{\tilde{\theta}}_tdt\right]\\
    &= \gamma w(\tau(0))+\gamma\int_0^{\tau^{-1}(T)} \sqrt{\dot{\tau}(t)}\underline{a}_t dt+\overline{\gamma}\int_0^T\underline{\tilde{a}}_t dt
    \\
&=\bar{S}_0-S_0-
\frac{\Upsilon'_0(\tau(0))}{\rho \Upsilon_0(\tau(0))}\Phi_0
\end{align}
where the last identity is just~\eqref{26}. It follows that the optimal initial turnover rate is
\begin{align}
    \hat{\phi}_0=\rho(\bar{S}_0-S_0)-\frac{\Upsilon'_0(\tau(0))}{ \Upsilon_0(\tau(0))}\Phi_0.
\end{align}
Invoking the same dynamic programming argument as in Lemma~3.5 of~\cite{BDR:22}, we conclude that the analogous formula holds for arbitrary $t$. This establishes optimality of the strategy described in Theorem~\ref{thm:1}.

\bibliographystyle{plainnat} \bibliography{finance}

\end{document}